\documentclass[aps,prl,twocolumn,showpacs,superscriptaddress]{revtex4-1} 
\usepackage[colorlinks=true,linkcolor=blue,citecolor=blue,urlcolor=blue]{hyperref}
\usepackage{dcolumn}   
\usepackage{bm}        
\usepackage{subfigure}
\usepackage{amsthm}
\usepackage[english]{babel}
\usepackage[]{units}
\usepackage{titlesec}
\usepackage{amsmath}
\usepackage{amsfonts}
\usepackage{amssymb}
\usepackage{braket}
\usepackage{lipsum}
\usepackage[pdftex]{color,graphicx}
\usepackage{times}

\newtheorem{Theorem}{Theorem}

\newtheorem{Lemma}{Lemma}

\usepackage[percent]{overpic}
\usepackage[export]{adjustbox}
\begin{document}

\title{Computational power of symmetry-protected topological phases}
\author{David T. Stephen}
\affiliation{Department of Physics and Astronomy, University of British Columbia, Vancouver, British Columbia V6T 1Z1, Canada}
\author{Dong-Sheng Wang}
\affiliation{Department of Physics and Astronomy, University of British Columbia, Vancouver, British Columbia V6T 1Z1, Canada}
\author{Abhishodh Prakash}
\affiliation{C. N. Yang Institute for Theoretical Physics and Department of Physics and Astronomy, State University of New York at Stony Brook, Stony Brook, NY 11794-3840, USA}
\author{Tzu-Chieh Wei}
\affiliation{C. N. Yang Institute for Theoretical Physics and Department of Physics and Astronomy, State University of New York at Stony Brook, Stony Brook, NY 11794-3840, USA}
\author{Robert Raussendorf}
\affiliation{Department of Physics and Astronomy, University of British Columbia, Vancouver, British Columbia V6T 1Z1, Canada}
\date{\today}

\begin{abstract}
We consider ground states of quantum spin chains with symmetry-protected topological (SPT) order as resources for measurement-based quantum computation (MBQC). We show that, for a wide range of SPT phases, the computational power of ground states is uniform throughout each phase. This computational power, defined as the Lie group of executable gates in MBQC, is determined by the same algebraic information that labels the SPT phase itself. We prove that these Lie groups always contain a full set of single-qubit gates, thereby affirming the long-standing conjecture that general SPT phases can serve as computationally useful phases of matter.

\end{abstract}
\maketitle
 
\textit{Introduction.} 
In many-body physics, the essential properties of a quantum state are determined by the phase of matter in which it resides. 
 Recent years have witnessed tremendous progress in the discovery and classification of quantum phases~\cite{Levin2005,Fidkowski2011,Chen2011,Schuch2011,
Chen2013,Barkeshli2014,Gu2014,Kapustin2015,Song2017,Lan2016}, and it is thus pertinent to ask---what can a phase of matter be used for? A traditional example is the ubiquitous superconductor, while newly discovered phases such as topological insulators~\cite{Qi2011} and quantum spin liquids~\cite{Balents2010} have promising future applications. 
Quantum phases are useful in quantum information processing as well: certain topological phases allow for error-resilient topological quantum computation via the braiding and fusion of their anyonic excitations~\cite{Kitaev2003,Nayak2008}. These applications all operate due to properties of a phase rather than a particular quantum state, hence they enjoy passive protection against certain sources of noise and error.
 
In this letter, we establish a general connection between the symmetry-protected topological (SPT) phases in one dimension (1D) ~\cite{Chen2011,Schuch2011,Chen2013} and quantum computation. To do this we use the framework of measurement-based quantum computation (MBQC)~\cite{Raussendorf2001,Raussendorf2003}, in which universal computation is possible using only single-body measurements on an entangled many-body system. The computational power of an MBQC scheme, defined by the set of logical gates that can be performed using measurements, is related to the entanglement structure of the many-body ground state. Whether this computational power is particular to individual states, or a property of a phase as above, is a long-standing open problem~\cite{Barrett2009,Doherty2009,Miyake2010,Bartlett2010,
Darmawan2012,Else2012,Miller2014,Prakash2014,Wei2014a,
Nautrup2015,Else2012a,Wang2016,Wei2017}.  An important early result showed that every ground state within certain SPT phases has the ability to faithfully transport quantum information along a 1D chain; however, ``universal'' single-qubit gates appeared to be properties only of special points in the phases~\cite{Else2012a}. Later, it was shown that, for one particular SPT phase (namely one that is protected by $S_4$ symmetry), universal single-qubit gates can be implemented throughout the entire phase~\cite{Miller2014}. Yet, it remains unknown whether a general SPT phase can serve as such a \textit{computational phase of matter}.

\begin{figure}
\centering
\includegraphics[width=\linewidth]{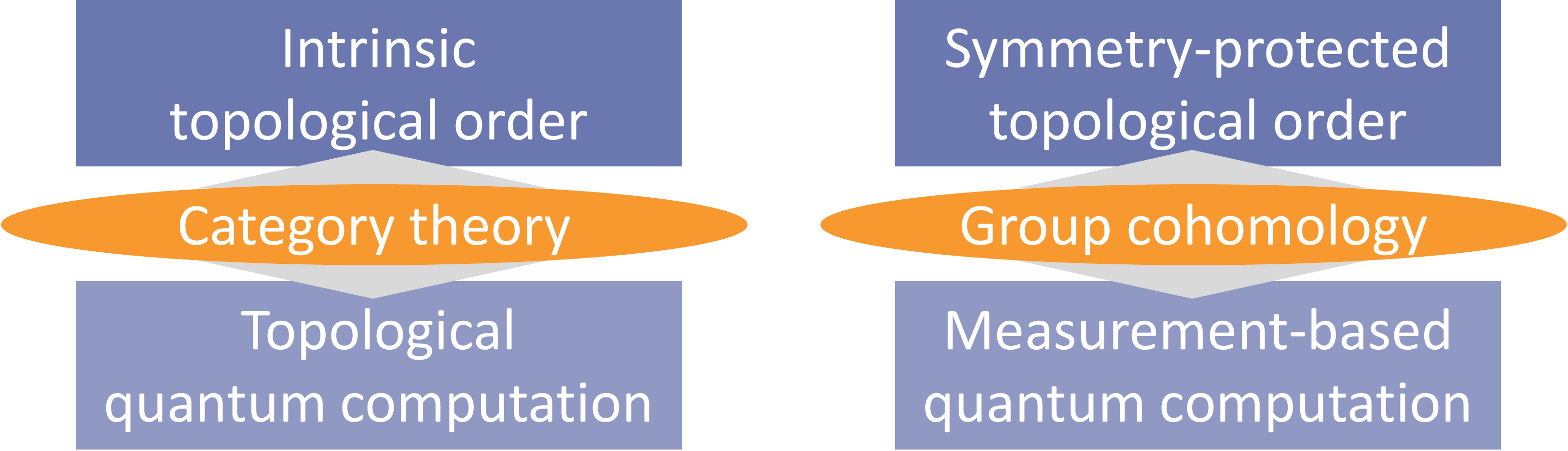}
\caption{In the same way that the language of category theory allows us to classify gates that can be executed by braiding the anyonic excitations of topologically ordered systems in 2D~\cite{Freedman2003,Nayak2008,Rowell2009}, group cohomology determines the gates implementable in measurement-based quantum computation using 1D resource states with symmetry-protected topological order.}
\label{fig:chart}
\end{figure}

Here, we construct a general computational scheme that harnesses the part of a ground state that is fully constrained by symmetry. This part is uniform throughout the SPT phase, and therefore the computational power in our scheme is a property of SPT phases rather than individual states. This power is determined by the same algebraic structure that is used to classify the SPT phases, namely group cohomology. This establishes a firm connection between SPT order and the computational power of many-body ground states. 

We can use this connection to prove that universal single-qubit gates are a property of all phases considered by Ref.~\cite{Else2012a}, and many more. Going beyond this, we identify classes of phases that also allow operations on qudits of arbitrarily large dimension. Overall, our results highlight how the algebraic classification of quantum phases can contribute to the study of the structures responsible for quantum computational power, as outlined in Fig.~\ref{fig:chart}

In the following, we begin by reviewing the virtual space picture of MBQC~\cite{Gross2007,*Gross2007a}, which aids our subsequent analysis. We then introduce the three key elements of our scheme, and demonstrate their use through the examples of the Affleck-Kennedy-Lieb-Tasaki (AKLT) state and the Haldane phase before generalizing to other SPT phases. We finish by using the algebraic classification of SPT phases to determine their computational power.

\textit{Computation in Virtual Space}.
We consider MBQC in the virtual space picture, where states are represented in the matrix product state (MPS) form~\cite{Perez-Garcia2006}. The wave function ${|\psi\rangle}$ of a 1D system of $N$ interacting sites of local dimension $d$ (i.e. a spin chain) can be written in MPS form by introducing the square matrices $A^i$, $i=0,\dots,d-1$, such that
\begin{equation} \label{eq:mps}
|\psi\rangle=\sum_{i_1,\dots,i_N} \langle R|A^{i_N}A^{i_{N-1}}\dots A^{i_1}|L\rangle {|i_1\dots i_N\rangle},
\end{equation}
where ${|R\rangle}$, ${|L\rangle}$ are states in the so-called ``virtual space'' that encode the boundary conditions of the finite chain. The MPS formalism leads to a useful interpretation of MBQC which occurs in virtual space~\cite{Gross2007,*Gross2007a}: measuring the leftmost spin in the chain with outcome ${|s\rangle}$ reduces chain length by one and evolves the virtual system as $|L\rangle\to \left(\sum_i \langle s|i\rangle A^i \right) |L\rangle$. With a proper choice of measurement basis, this can correspond to unitary evolution and can simulate computation up to outcome-dependent byproduct operators. Since we consider only 1D resource states, we say a state is \emph{universal} if measurement can induce a full set of gates for a single qudit, corresponding to operators in $SU(D)$ on some $D$-level subspace in virtual space.

A simple example is the spin-1 AKLT state, which is well-known to be a universal resource~\cite{Brennen2008}. The MPS matrices are the Pauli matrices, $A^i=\sigma^i$, with respect to the \textit{wire basis} $\mathcal{B}=\left\{|x\rangle,{|y\rangle},{|z\rangle}\right\}$ where ${|i\rangle}$ is the 0 eigenstate of the spin-1 operator $S^i$. To achieve a rotation by $\theta$ about the $z$-axis, we measure in the basis $\mathcal{B}(z,\theta)=\left\{{|\theta_x\rangle},{|\theta_y\rangle},{|z\rangle}\right\}\equiv\left\{\cos\frac{\theta}{2}{|x\rangle}-\sin\frac{\theta}{2}{|y\rangle},\sin\frac{\theta}{2}{|x\rangle}+\cos\frac{\theta}{2}{|y\rangle},{|z\rangle}\right\}$ and propagate the byproducts $\sigma^x$, $\sigma^y$, and $\sigma^z$, respectively. We enact byproduct propagation via symmetry transformations of future measurement bases, as described in Ref.~\cite{Else2012}. With this, the first two outcomes give the desired rotation by $\theta$ while the third does nothing, so the gate is probabilistic with success probability $\frac{2}{3}$. Rotations about the $x$-axis can be achieved similarly, giving a full set of $SU(2)$ operations. 

To extend the universality of the AKLT state and others like it to entire SPT phases, we introduce three modifications to the usual MBQC procedure, as described in Fig.~\ref{fig:main}. The purpose and justification of each are given in the following section, using the AKLT state and Haldane phase as examples.

\textit{Computation in the Haldane phase}.
We begin this section by introducing the ``mixed state interpretation'' of MBQC that will be used throughout this letter. Here we argue its validity, with a formal proof given in the Supplementary Material. We define a computation by a sequence of $n$ measurement bases, which are fixed modulo byproduct propagation. In general, an input state $|\psi\rangle$ will be taken to a final state ${|\psi_{\vec{s}}\rangle}$ which depends on the measurement outcomes $\vec{s}=(s_1,\dots,s_n)$. Then we measure some observable $O$ on ${|\psi_{\vec{s}}\rangle}$, whose eigenvalues $o_i$ appear with probability $p(o_i|\vec{s})$. To garner measurement statistics of $O$, we must repeat the computation, whereupon the full statistics are given by $p(o_i)=\sum_{\vec{s}} p(o_i|\vec{s})p_{\vec{s}}$ where $p_{\vec{s}}$ is the probability of outcomes $\vec{s}$. These statistics are encoded in the mixed state $\hat{\sigma}=\sum_{\vec{s}} p_{\vec{s}}{|\psi_{\vec{s}}\rangle}{\langle \psi_{\vec{s}}|}$, for instance $\langle O \rangle=\sum_{\vec{s}} p_{\vec{s}}{\langle \psi_{\vec{s}}|}O{|\psi_{\vec{s}}\rangle}\equiv\text{Tr}\,     (O\hat{\sigma})$. Hence in this probabilistic scenario the computational output must be interpreted to be $\hat{\sigma}$.

To determine the mixed state $\hat{\sigma}$, we simply sum over all possible outcomes of each measurement. It is crucial that this sum-over-outcomes is implemented \textit{after} byproduct propagation, making it very different from simply tracing over each spin in the chain. The byproducts accumulated at the end of the computation affect the basis of computational readout, during which we do not sum over outcomes. By analysing the computation in this way, we can design a sequence of measurement bases such that $\hat{\sigma}$ approximates the desired output. If the computation defined by this sequence of measurements is repeated many times, it deterministically produces the desired measurement statistics of any observable $O$, even though each run of the algorithm may produce a different output state that is meaningless on its own. 

\begin{figure} 

\includegraphics[width=\linewidth]{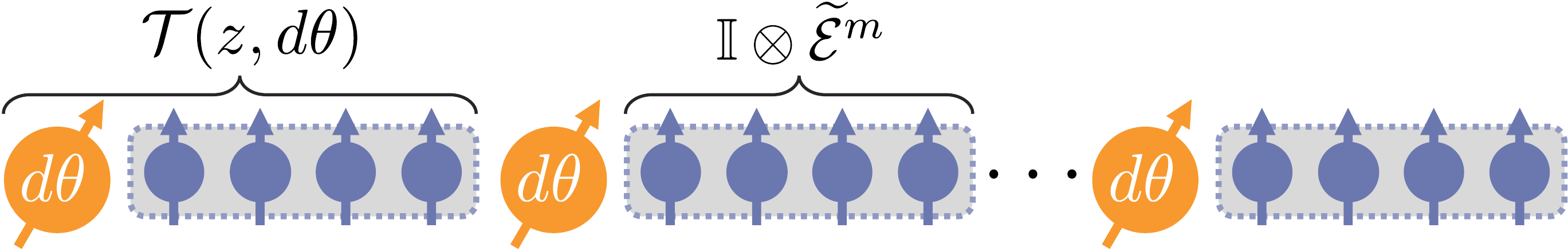}
\caption{Illustration of the measurements needed to execute a rotation about the $z$-axis in the Haldane phase example. Our scheme consists of three modifications to the usual MBQC procedure: (1) In analysis of the scheme, measurement outcomes are summed over, such that the computational output is interpreted as a mixed state, (2) finite rotations are split into smaller pieces $d\theta$ that each differ only slightly from the identity, and (3) consecutive gates are separated by many applications of the identity gate.}
\label{fig:main}
\end{figure}

Let us return to the AKLT state as an example. By measuring in the basis $\mathcal{B}(z,\theta)$ and summing over measurement outcomes, we find that an initial state ${|L\rangle}{\langle L|}$ becomes:
\begin{equation}
\hat{\sigma}= \frac{2}{3} e^{-i\theta\sigma_z /2}{|L\rangle}{\langle L|} e^{i\theta\sigma_z /2} +\frac{1}{3}{|L\rangle}{\langle L|}.
\end{equation}
Since the original gate is probabilistic, this is a mixed state and does not represent unitary evolution. However, for small angles $d\theta$, it is unitary up to first order:
\begin{equation}\label{eq:akltmix}
\hat{\sigma}=e^{-i\frac{2}{3}d\theta\sigma_z /2}{|L\rangle}{\langle L|} e^{i\frac{2}{3}d\theta\sigma_z /2} + \mathcal{O}(d\theta^2).
\end{equation}

So for small rotation angles $d\theta$, the mixed output state is our initial state rotated by a reduced angle $\frac{2}{3}d\theta$ about the $z$-axis. Restriction to gates that are close to the identity is an unavoidable consequence of the mixed state interpretation, and finite rotations must be split into many infinitesimal pieces~\footnote{Each of the infinitesimal rotations that compose a finite gate must be implemented one after another while propagating byproducts.}. The number of measurements needed to execute a unitary gate with rotation angle $\theta$ and admissible error $\epsilon$ is $\mathcal{O}(\theta^2/\epsilon)$; details can be found in the Supplementary Material.

The AKLT state is in the Haldane phase, which we define as the SPT phase protected by on-site $\mathbb{Z}_2\times\mathbb{Z}_2$  symmetry~\cite{Pollman2012}.
 Every state in the Haldane phase can be viewed as an AKLT state with some additional entanglement that encodes the microscopic details of the state. This is formally expressed in terms of the MPS matrices, which factorize as $A^i=
\sigma^i \otimes B^i$ in the wire basis $\mathcal{B}$~\cite{Else2012a}. The Pauli part acts in the \textit{logical subspace} into which information is encoded and processed. The matrices $B^i$ act in the \textit{junk subspace} and contain all of the microscopic details of the state. Importantly, byproduct propagation via symmetry transformations acts only within the logical subspace. This is not a problem for measurements in the wire basis, which evolve the two subsystems independently. But measurements in other bases will mix the junk and logical subspaces, which hides the logical information and introduces an unavoidable outcome dependence into the computation. We now show how the mixed state interpretation allows us to solve both of these problems in a relatively simple way. 

Consider a measurement in the infinitesimally tilted basis $\mathcal{B}(z,d\theta)$. Without loss of generality, we assume that our initial state is factorized across the subspaces as ${|\phi\rangle}{\langle \phi|}\otimes \rho_{\rm{fix}}$, for a particular fixed point state $\rho_{\rm{fix}}$ that will be defined later. If we get the outcome ${|\theta_x\rangle}$ and propagate $\sigma^x$ on the logical subspace, our state becomes:
\begin{align} \label{eq:d1}
&{|\phi\rangle}{\langle \phi|}\otimes \rho_{\rm{fix}}\to 
{|\phi\rangle}{\langle \phi|}\otimes B^x\rho_{\rm{fix}}B^{x\dagger}  \\
&+i\frac{d\theta}{2}\left({|\phi\rangle}{\langle \phi|}\sigma^z\otimes B^x\rho_{\rm{fix}} B^{y\dagger}-\sigma^z{|\phi\rangle}{\langle \phi|}\otimes B^y\rho_{\rm{fix}} B^{x\dagger}\right), \nonumber
\end{align}
up to first order in $d\theta$. We see that the two subsystems are no longer factorized, and the logical state ${|\phi\rangle}{\langle \phi|}$ is no longer accessible. 

To remedy this, we will flow the junk subspace towards a fixed point. This is accomplished by simply measuring a large number of spins in the wire basis. In the mixed state interpretation, a measurement in the wire basis followed by logical byproduct propagation effects the operation $ \mathbb{I}\otimes \sum_i B^i(\cdot)B^{i\dagger}\equiv \mathbb{I}\otimes \widetilde{\mathcal{E}}$. Since every state in the Haldane phase is short-range correlated, the channel $\widetilde{\mathcal{E}}$ will have a unique fixed point, which is $\rho_{\rm{fix}}$, with all other eigenvalues of modulus less than unity \cite{Else2012}. Hence measuring $m$ consecutive spins in the wire basis results in the linear channel $\mathbb{I}\otimes \widetilde{\mathcal{E}}^{m}$ and projects the junk subspace onto the fixed point $\rho_{\rm{fix}}$. The projection occurs exponentially fast over the correlation length $\xi$ of the state.

Applying this to Eq.~\ref{eq:d1}, which must be summed with its counterparts for the other measurement outcomes ${|\theta_y\rangle}$ and ${|z\rangle}$, we find that for large enough $m$, 
\begin{equation}
\hat{\sigma}=\left(\nu {|\phi\rangle}{\langle \phi|}+i\frac{d\theta}{2}(\nu_{xy}+\nu_{yx})\left[{|\phi\rangle}{\langle \phi|},\sigma^z\right]\right)\otimes \rho_{\rm{fix}},
\end{equation}
where we have defined $\lim_{m\to\infty} \widetilde{\mathcal{E}}^{m}(B^{i}\rho_{\rm{fix}}B^{j\dagger})=\nu_{ij}\rho_{\rm{fix}}$ and $\nu=\nu_{xx}+\nu_{yy}+\nu_{zz}$. Up to first order in $d\theta$, this corresponds to a unitary rotation acting on the logical subspace:
\begin{equation}
\mathcal{T}(z,d\theta)=\exp\left\{-id\theta\left(\frac{\nu_{xy}+\nu_{yx}}{2\nu}\right)\sigma^z\right\}.
\end{equation}

Hence, making a measurement in the rotated basis $\mathcal{B}(z,d\theta)$, followed by a series of measurements in the wire basis, produces the desired rotation of the virtual state ${|\phi\rangle}$ up to a scaling factor $\frac{\nu_{xy}+\nu_{yx}}{\nu}$. As long as this factor is non-zero, it can be measured on the chain prior to computation by attempting a finite rotation (split into small pieces), and measuring the reduction in rotation angle~\cite{Raussendorf2017}. The parameters $\nu_{ij}$ contain all relevant microscopic details of our resource state ${|\psi\rangle}$. Since they can be measured during a calibration step, any state in the phase can be used as a resource without prior knowledge of its identity.  

We can repeat the above procedure for rotations about the $x$-axis to generate all of $SU(2)$. Hence every state in the Haldane phase, with the exception of a null subset in which some of the constants $\nu_{ij}$ are 0, has the same computational power as the AKLT state (which satisfies $\nu_{ij}=\frac{1}{3}\ \forall i,j$). To complete the scheme, we would require a method to read out and initialize the virtual state which also works throughout the phase. This can be done without the need of ancillary systems on the boundaries~\cite{Raussendorf2017}.

\textit{Generalization to Other Phases}. 
Our scheme does not depend on any properties that are particular to the Haldane phase, so it can be generalised to a large class of other SPT phases.
A general 1D SPT phase without symmetry breaking is defined with respect to an on-site symmetry group $G$ such that $u(g)^{\otimes N}{|\psi\rangle}={|\psi\rangle}$ for some unitary representation $u$ of $G$. The phase is then labelled by a cohomology class $[\omega]\in H^2(G,U(1))$ in the second cohomology group of $G$ which describes how this symmetry acts in the virtual space~\cite{Chen2011}. 

The Haldane phase is an example of a maximally non-commutative SPT phase, as defined in Ref.~\cite{Else2012a}. Such phases satisfy all conditions needed to apply our methods, namely the existence of a logical subspace and the ability to propagate byproduct operators within it. Indeed, suppose that $G$ is finite abelian and $[\omega]$ is maximally non-commutative, meaning $\{g\in G|\omega (g,g')=\omega (g',g)\  \forall g'\in G\}=\{e\}$. By diagonalizing the representation $u$, we obtain the wire basis $\mathcal{B}=\{{|0\rangle},\dots,{|d-1\rangle}\}$ such that $u(g){|i\rangle}=\chi_i(g){|i\rangle}\ \forall g\in G$ where $\chi_i(g)$ are linear characters of $G$. Maximal non-commutativity then implies the MPS tensor $A^i$ can be written in the wire basis as~\cite{Else2012a}:
\begin{equation}\label{eq:decomp}
A^i=C^i\otimes B^i,
\end{equation}

\noindent where $C^i$ are $D\times D$ unitary  and trace-orthogonal matrices and $D=\sqrt{|G|}$ is the dimension of our logical subspace~\footnote{$D$ is always defined since any finite abelian group which supports a maximally non-commutative factor system must have the form $G\cong G'\times G'$ for some subgroup $G'$~\cite{Berkovich1998}.}. $C^i$ can be determined uniquely from $G$, $[\omega]$, and $\chi_i$ as described in the Supplementary Material. In general, if some group $G$ has a finite abelian subgroup $H$ such that $[\omega|_H]$ is maximally non-commutative, we can make the exact same argument with $H$ taking the place of $G$ everywhere. This means the following results also apply to certain non-abelian groups and Lie groups.

Now we follow the same steps used to perform computation in the Haldane phase. Measurement in the slightly tilted basis $\mathcal{B}(i,j;d\theta,\varphi)=$ $\{ {|0\rangle},\dots,{|i\rangle}+d\theta e^{i\varphi}{|j\rangle},$ ${|j\rangle}-d\theta e^{-i\varphi}{|i\rangle},\dots,{|d-1\rangle}\}$, followed by measurements in $\mathcal{B}$ to drive the junk subspace to a fixed-point state, induces an infinitesimal rotation in the logical subspace:
\begin{align} \label{eq:gates}
&\mathcal{T}(i,j;d\theta,\varphi)= \\ \nonumber &\exp\left\{d\theta \frac{|\nu_{ij}|}{\nu}\left(e^{i(\varphi+\delta_{ij})}C^{i\dagger}C^j-e^{-i(\varphi+\delta_{ij})}C^{j\dagger}C^i\right)\right\},
\end{align}
where $\nu_{ij}=|\nu_{ij}|e^{i\delta_{ij}}$ is as defined earlier and $\nu=\sum_{i=0}^{d-1}\nu_{ii}$. As before, the microscopic details of the state enter only as these measurable constants. Computation can only proceed if these constants are non-zero, which is satisfied for all but a null set of states. With knowledge of these constants, $\mathcal{B}(i,j;d\theta,\varphi)$ can be chosen such that the primitive gates are generated by elements of the set of anti-hermitian operators:
\begin{equation} 
\mathcal{O}=\left\{\alpha C^{i\dagger}C^j-\alpha^*C^{j\dagger}C^i \right\}
\end{equation}
with $i,j=0\dots d-1$, $i\neq j$, $|\alpha|\ll 1$. Furthermore, we have $e^{d\theta A}e^{d\theta B}e^{-d\theta A}e^{-d\theta B}\approx e^{(d\theta)^2[A,B]}$, so that our infinitesimal generators form a real Lie algebra which in turn generates a Lie group $\mathcal{L}[\mathcal{O}]$ of executable gates.

From the above, we can see the main strength of our methods. Given only the algebraic quantities $G$, $u$, and $[\omega]$ which describe the SPT phase of our resource state, we are able to define a complete MBQC scheme, including the set of gates and the measurements needed to execute them. The computational power of each state in the phase is uniformly defined as the Lie group $\mathcal{L}[\mathcal{O}]$, which is completely determined by the same algebraic quantities. This signifies the existence of a deep connection between SPT order and MBQC via the language of group cohomology.

\textit{Determining Computational Power}.
To determine the computational power of a phase, we must identify the Lie group $\mathcal{L}[\mathcal{O}]$. We will do this by taking advantage of the algebraic structure inherited from the SPT phase classification. Consider first the case where the representation $u|_H$ contains all non-trivial characters of the subgroup $H$. This means that $\mathcal{O}$ contains $D^2-1$ trace-orthogonal, antihermitian operators, so $\mathcal{L}[\mathcal{O}]\cong SU(D)$. If the Hilbert space dimension of our physical sites is smaller than $D^2-1$, or certain characters $\chi_i$ do not appear in $u|_H$, $\mathcal{L}[\mathcal{O}]$ may be some Lie subgroup of $SU(D)$. However, with the condition of maximal non-commutativity, this subgroup is always universal on a qudit system, as stated in the following theorem:

\begin{Theorem}
Consider an SPT phase defined by an on-site symmetry group $G$ and cohomology class $[\omega]$. Suppose there exists a finite abelian subgroup $H\subset G$ such that $[\omega|_H]$ is maximally non-commutative, and let $p^n$ be a prime power dividing $\sqrt{|H|}$. Then $\mathcal{L}[\mathcal{O}]\supset SU(p^n)$.
\end{Theorem}

\noindent 
This result, proven in the Supplementary Material, determines the\textit{ minimal} computational power of the phase, which is \textit{independent} of $u$ and hence uniform amongst the phase. This shows that 1D ground states with SPT order are generically useful as MBQC resources.

Beyond this minimal case, $\mathcal{L}[\mathcal{O}]$ can often be expanded to gain additional computational power. For example when $H=(\mathbb{Z}_2)^4$, our theorem guarantees that $SU(2)\subset \mathcal{L}[\mathcal{O}]$, but this can be expanded to either $SU(4)$ or $SU(2)\times SU(2)$ depending on the on-site symmetry representation $u$. So, while changing $u$ is generally considered to not change the SPT phase of a system~\cite{Schuch2011}, it remains an important label for total computational power in our scheme. If, however, we allow ourselves to redefine the locality of measurements by blocking neighbouring sites, $\mathcal{L}[\mathcal{O}]$ will always equal $SU(D)$ after sufficient blocking.

Now we must ask: which symmetry groups protect phases that satisfy our theorem? To answer this in general is a difficult problem of group cohomology, but we can identify some particularly relevant examples. When $G$ is a classical Lie group (except $Spin(4n)$), there is a subgroup of the form $\mathbb{Z}_N\times\mathbb{Z}_N\subset G$ such that $H^2(G,U(1))\cong H^2(\mathbb{Z}_N\times\mathbb{Z}_N,U(1))$ \cite{Duivenvoorden2013,Duivenvoorden2013a}. Since $\mathbb{Z}_N\times\mathbb{Z}_N$ protects a maximally non-commutative phase~\cite{Berkovich1998,Else2012a}, $G$ must protect a phase which satisfies our theorem. The same can be said for any subgroup $G'$ such that $\mathbb{Z}_N\times\mathbb{Z}_N\subset G'\subset G$. This has already been observed in Ref.~\cite{Prakash2014} for the groups $D_4,A_4,S_4\subset SO(3)$, which each contain $\mathbb{Z}_2\times\mathbb{Z}_2$. Another example is the class of groups for which the subgroup $H$ specified in Theorem 1 appears as a (semi)direct factor, that is $G=H'\rtimes H$ for some subgroup $H'$ which could represent eg. time reversal symmetry~\cite{Xiong2016}.

\textit{Conclusion}. 
By introducing three simple modifications to the usual MBQC procedure, we showed that the MBQC power of an SPT-ordered ground state of a spin chain is determined solely by the cohomological information that labels the corresponding SPT phase, and that this power is always sufficient for universal computation on a single qudit. Regarding the algebraic classification of phases of matter and its role in quantum computation, our results show that group cohomology links SPT order and MBQC in 1D, in the same way that modular tensor categories link topological order and topological quantum computation in 2D~\cite{Freedman2003,Nayak2008,Rowell2009}. In each case, the algebraic framework that classifies the phases of matter also classifies their computational properties. Whether this extends to higher dimensions and other types of quantum phases is an intriguing question at the intersection of quantum information and condensed matter physics. There is already evidence that SPT order in higher dimensions can lead to unique computational properties~\cite{Miller2015,Miller2016,Yoshida2015,Yoshida2016,
Nautrup2015,Darmawan2012,Wei2017}. It would also be interesting to see whether the mathematical frameworks that unify topological order and SPT order, such as G-crossed braided tensor categories~\cite{Barkeshli2014}, could also describe computation with systems that have both types of order.

\begin{acknowledgments}
This work is supported by NSERC, the Canadian Institute for Advanced Research (Cifar), and the National Science Foundation under Grant No. PHY 1620252. R.R. is a fellow of the Cifar Quantum Information Science program.
\end{acknowledgments}

\appendix
\section{Supplementary Material}

Here we present the supplementary material for the main text. We begin by proving the validity of the mixed state interpretation used in our methods. We then present an analysis of the error in our scheme and the associated cost of computation. Finally, we prove the main theorem of the text after giving relevant information on the cohomological description of SPT phases.

\section{Mixed State Interpretation}

Here we argue the validity of our mixed state interpretation and the corresponding ``sum over outcomes'' approach. Consider an MBQC scheme in which the logical evolution depends on measurement outcomes, even when supplemented by byproduct propagation. That is, we assume that measuring the next spin in the chain with outcome ${|s\rangle}$ enacts the evolution

\begin{equation}
{|L\rangle}\to \frac{1}{\sqrt{p_s}}\left(\sum_i  {\langle s|i\rangle}A^i\right){|L\rangle}=\frac{1}{\sqrt{p_s}}\Sigma_s\Gamma_s{|L\rangle},
\end{equation}

\noindent where $\Sigma_s$ is the unitary byproduct operator, and $\Gamma_s$ is the desired evolution, which at this point may not yet be unitary. $p_s$ is the probability of obtaining outcome $s$, which appears via the Born rule.

Now a general computation involves the measurement of $m$ spins in any basis with outcomes $\vec{s}=(s_1,\dots,s_m)$, propagating byproduct operators after each step. The initial state ${|L\rangle}$ evolves to a final state $\frac{1}{\sqrt{p_{\vec{s}}}}\Sigma_{\vec{s}}{|L'_{\vec{s}}\rangle}$ where ${|L'_{\vec{s}}\rangle}=\Gamma_{s_m}\dots\Gamma_{s_1}{|L\rangle}$, $\Sigma_{\vec{s}}=\prod_{i=1}^m \Sigma_{s_i}$ is the accumulated byproduct operator, and $p_{\vec{s}}$ is the probability of the outcome string $\vec{s}$. At this point, our resource state ${|\psi\rangle}$ has evolved to

\begin{equation}
{|\psi_{\vec{s}}\rangle}=\frac{1}{\sqrt{p_{\vec{s}}}}\sum_{i_{m+1}\dots i_n} {\langle R|}A^{i_n}\dots A^{i_{m+1}}\Sigma_{\vec{s}}{|L'_{\vec{s}}\rangle}{|i_{m+1}\dots i_n\rangle}.
\end{equation}

Computation ends with readout of some observable $O$ on the final state ${|L'_{\vec{s}}\rangle}$. Our only tool available to do this is a measurement of the next spin in the chain, $m+1$, whose measurement outcome must be used to infer something about $O$. Let $\{{|o_i\rangle}|i=0,\dots,d-1\}$ be the relevant measurement basis for read out of $O$, which has been appropriately modified to propagate the accumulated byproduct operator $\Sigma_{\vec{s}}$ past the readout site. Letting $A[o_\alpha]=\sum_i {\langle o_\alpha|i\rangle}A^i$ we can calculate the probability of obtaining the outcome ${|o_\alpha\rangle}$ given the previous outcomes $\vec{s}$:

\begin{widetext}

\begin{align}
p(o_\alpha|\vec{s})&={\langle o_\alpha|\psi_{\vec{s}}\rangle}{\langle \psi_{\vec{s}}|o_\alpha\rangle} \nonumber \\
&=\sum_{i_{m+2}\dots i_n}{\langle R|}A^{i_n}\dots A^{i_{m+2}}\Sigma_{\vec{s}}A[o_\alpha]\frac{{|L'_{\vec{s}}\rangle}{\langle L'_{\vec{s}}|}}{p_{\vec{s}}}A[o_\alpha]^\dagger\Sigma_{\vec{s}}^\dagger A^{i_{m+2}\dagger}\dots A^{i_n\dagger}{|R\rangle} \nonumber \\
&=\sum_{i_{m+2}\dots i_n}\text{Tr}\, \left[(A^{i_{k+1}\dagger}\dots A^{i_n\dagger}{|R\rangle}{\langle R|}A^{i_n}\dots A^{i_{k+1}}) (A^{i_k}\dots A^{i_{m+2}}\Sigma_{\vec{s}}A[o_\alpha]\frac{{|L'_{\vec{s}}\rangle}{\langle L'_{\vec{s}}|}}{p_{\vec{s}}}A[o_\alpha]^\dagger\Sigma_{\vec{s}}^\dagger A^{i_{m+2}\dagger}\dots A^{i_k\dagger})\right].
\end{align}

\end{widetext}

Where $k$ is taken such that $m\ll k\ll n$.  Now the first term in the trace can be rewritten as $(\mathcal{E}^{\dagger})^{n-k}({|R\rangle}{\langle R|})$ where $\mathcal{E}^\dagger=\sum_i A^{i\dagger}(\cdot)A^i$ is the adjoint of the usual MPS channel. By the injectivity and the canonical form, $\mathcal{E}^\dagger$ has a unique fixed point $\Lambda$ with all other eigenvalues of modulus less than unity~\cite{Perez-Garcia2006}. Now specializing to the relevant case where we have a wire basis in which $A^i=C^i\otimes B^i$, it can be shown that $\Lambda=\frac{1}{D}\mathbb{I}\otimes \widetilde{\Lambda}$ for some density matrix $\widetilde{\Lambda}$~\cite{Else2012}. Then we can rewrite $(\mathcal{E}^{\dagger})^{n-k}({|R\rangle}{\langle R|})=\mathbb{I}\otimes \widetilde{\Lambda}$. By cyclicity of the trace, we can eliminate the unitary parts $C^i$ of the MPS matrices in the second term of the trace. Furthermore, the byproduct $\Sigma_{\vec{s}}$ acts only on the logical subspace, so it can be eliminated in the same way. We are left with:

\begin{equation}
p(o_\alpha|\vec{s})=\text{Tr}\, \left[(\mathbb{I}\otimes \widetilde{\Lambda})(\mathbb{I}\otimes \widetilde{\mathcal{E}}^k)\left(A[o_\alpha]\frac{{|L'_{\vec{s}}\rangle}{\langle L'_{\vec{s}}|}}{p_{\vec{s}}}A[o_\alpha]^\dagger\right) \right],
\end{equation}

\noindent where $\widetilde{\mathcal{E}}=\sum_i B^i(\cdot)B^{i\dagger}$ is the MPS channel associated to the junk space as introduced in the main text. Importantly, we notice that the probability does not depend on the accumulated byproduct operator $\Sigma_{\vec{s}}$ so long as the measurement is done sufficiently far from the boundary. Now, the total probability of obtaining outcome ${|o_\alpha\rangle}$ is given by $p(o_\alpha)=\sum_{\vec{s}} p(o_\alpha|\vec{s})p_{\vec{s}}$. By exploiting linearity within the above expression, we have:

\begin{equation} \label{eq:prob}
p(o_\alpha)=\text{Tr}\,      \left[(\mathbb{I}\otimes \widetilde{\Lambda})(\mathbb{I}\otimes \widetilde{\mathcal{E}}^k)(A[o_\alpha]\hat{\sigma}A[o_\alpha]^\dagger)\right],
\end{equation}

\noindent where we have introduced the mixed state $\hat{\sigma}$ given by:

\begin{align}
\hat{\sigma}&=\sum_{\vec{s}} {|L'_{\vec{s}}\rangle}{\langle L'_{\vec{s}}|} \nonumber \\
&=\sum_{s_m}\Gamma_{s_m}\left(\ \dots \ \left(\sum_{s_1} \Gamma_{s_1}{|L\rangle}{\langle L|}\Gamma_{s_1}^\dagger \right)\  \dots \ \right) \Gamma_{s_m}^\dagger.
\end{align}

So we see that the readout statistics of $O$ are encoded in the mixed state $\hat{\sigma}$ which can be determined by summing over the outcomes of each measurement during computation. This sum occurs after byproduct propagation, and the accumulated byproduct operator has no effect other than changing the final readout basis.

\section{Error and Cost Analysis}

Here we address the question of what the computational cost is to implement a unitary gate with rotation angle $\theta$ while allowing for an error $\epsilon$. The basic argument is that if a rotation about an angle $\theta$ is subdivided into $N$ rotations about an angle $\theta/N$, then the error per individual small rotation is of order $(\theta/N)^2$, and the cumulative error over $N$ such rotations is thus $\epsilon_N = {\cal{O}}(1/N)$. Hence, to get by with a total error of $\epsilon$, we require a subdivision of the rotation into $N\sim 1/\epsilon$ steps.

In more detail, there are two sources of error in the present construction for unitary gates. First, an elementary unitary operation with small rotation angle $d\theta$ incurs an error at second order in $d\theta$, as discussed above. Second, every individual rotation $\mathcal{T}(d\theta)$ requires the junk system to be brought into the fixed point state $\rho_\text{fix}$. This could be achieved by measuring an infinite number of spins in the wire basis. In any reasonable implementation, we measure only a finite number of spins, producing an error in the state of the junk system compared to the true fixed point state. Fortunately, this error is exponentially small in the number $n$ of measured spins,
$$
\| \rho_\text{junk} - \rho_\text{fix}\| \sim \exp(-n/\widetilde{\xi}),
$$
where $\widetilde{\xi}$ is a correlation length $\widetilde{\xi}=-\ln(\lambda_1)$ associated to the junk subsystem, with $\lambda_1$ the second-largest eigenvalue of the channel ${\widetilde{\cal{E}}}$. This correlation length is less than or equal to the true correlation length $\xi$ of our resource state.

The cumulative error due to imperfect preparation of the fixed point state of the junk system is $\epsilon_\text{fix}={\cal{O}}\left(N {\lambda_1}^n\right)$. Therefore, the choice
\begin{equation}\label{Efix}
n = 2\widetilde{\xi} \log N 
\end{equation}
leads to an error $\epsilon_\text{fix}={\cal{O}}(1/N)$, which is the same scaling as for $\epsilon_N$.

Splitting the total error $\epsilon$ evenly between $\epsilon_\text{fix}$ and $\epsilon_N$, $\epsilon_\text{fix}=\epsilon_N =\epsilon/2$, we find that we can achieve a total error of $\epsilon$ for the choice 
$$
N \sim \frac{1}{\epsilon}.
$$
With Eq.~(\ref{Efix}), the total cost of implementing a unitary gate within an error $\epsilon$, in terms of total number $C=N(n+1)$ of local measurements, is thus
\begin{equation}
C = {\cal{O}}\left(\frac{\widetilde{\xi}}{\epsilon} \log\left(\frac{1}{\epsilon}\right)\right).
\end{equation}

\section{Proof of Theorem}

Here we prove the main theorem of the text, repeated here for convenience.

\setcounter{Theorem}{0}

\begin{Theorem}
Consider an SPT phase defined by an on-site symmetry group $G$ and cohomology class $[\omega]$. Suppose there exists a finite abelian subgroup $H\subset G$ such that $[\omega|_H]$ is maximally non-commutative, and let $p^n$ be a prime power dividing $\sqrt{|H|}$. Then $\mathcal{L}[\mathcal{O}]\supset SU(p^n)$.
\end{Theorem}

Our proof utilizes a graphical description, which facilitates the proof of our theorem. It also helps to determine how to most efficiently generate any gate by composing the primitive gates.

\subsection{Maximally Non-commutative Phases}

Before proving this theorem, we give the necessary background information. Consider a matrix product state in an SPT phase labelled by a group $G$ and a cohomology class $[\omega]\in H^2(G,U(1))$. Given a linear representation $u$ of $G$ acting on the physical spins, our MPS tensors satisfy the following symmetry condition:

\begin{equation}
\sum_{j} u(g)_{ij} A^j=V(g)^\dagger A^i V(g), \ \forall g\in G.
\end{equation}

\noindent Therein, the projective representation $V(G)$ can be written $V(g)=\bigoplus_\alpha \mathbb{I}_\alpha\otimes  V_\alpha (g)$. $V_\alpha (G)$ are all of the irreducible representations of $G$ with $\omega$ as their factor system. By considering a finite abelian subgroup $H$ of $G$, we have:

\begin{equation}
u(h)=\bigoplus_{i=0}^{d-1} \chi_i (h)\ \forall h\in H,
\end{equation}

\noindent where each $\chi_i$ is a linear character of $H$. 
We require the following Lemma on projective representations:

\begin{Lemma}
Consider an irreducible projective representation $V$ of an abelian group $H$ with factor system given by the cocycle $\omega$ ($V(h)V(h')=\omega(h,h')V(h')V(h)$). The following conditions are equivalent:

1) $[\omega]$ is maximally non-commutative. That is, $\{h\in H|\omega (h,h')=\omega (h',h) \forall h'\in H\}=\{e\}$.

2) $V$ is the unique projective irrep with cocycle $\omega$, and $V$ has dimension $\sqrt{|H|}$.

3) $\text{Tr}\,       V(h)=\sqrt{|H|}\delta_{h,e}$
\end{Lemma}

\begin{proof}
1) = 2) follows from Theorem 6.39 of \cite{Berkovich1998}.

3)$\to$ 2). If $\text{Tr}\,       V(e)=\text{Tr}\,      I=\sqrt{|H|}$, then $V(H)$ has degree $\sqrt{|H|}$ and furthermore it is the unique $\omega$-irrep by Theorem 6.13 of \cite{Berkovich1998}.

1),2)$\to$ 3). If $V(H)$ is maximally non-commutative, then $\{h\in H|V(h)=\lambda \mathbb{I}\}$ is trivial. Then, since the degree of $V(H)$ is $\sqrt{|H|}$, we have 3) by Corollary 1.11.13 of \cite{Karpilovsky1994}.

\end{proof}
Now suppose $[\omega|_H]$ is maximally non-commutative. Then, using Lemma 1 and the methods of Ref.~\cite{Else2012a}, we can show that any MPS in this phase satisfies:

\begin{equation}
A^i= V(h_i)\otimes B^i,
\end{equation}

\noindent where $V$ is the unique irrep with factor system $\omega|_H$, and $h_i$ is uniquely defined by the relation $V(h_i)V(h)=\chi_i(h)V(h)V(h_i)$ $\forall h\in H$. So we have the required virtual space decomposition $A^i=C^i\otimes B^i$ where $C^i=V(h_i)$. We also have $\sum_j u(h)_{ij}A^j=(V(h)^\dagger\otimes \mathbb{I})A^i(V(h)\otimes \mathbb{I})$, which allows us to propagate all byproducts in the logical subspace using symmetry transformations. 

By Lemma 36 of Ref.~\cite{Berkovich1998} and its proof within, if $[\omega|_H]$ is maximally non-commutative then $H$ must have the form $H_1\times\dots\times H_r$ where $H_i\cong \mathbb{Z}_{p_i^{n_i}}\times\mathbb{Z}_{p_i^{n_i}}$ and $p_i^{n_i}$ is a prime power. Furthermore, $[\omega|_{H_i}]$ is also maximally non-commutative for all subgroups $H_i$. 
By restricting to any such subgroup $\tilde{H}=\mathbb{Z}_D\times\mathbb{Z}_D$ for $D=p^n$, the operators $C^i$ can be taken to be Heisenberg-Weyl operators of the form $Z^iX^j$ where $XZ=\Omega ZX$ and $\Omega=e^{\frac{2\pi i}{D}}$. In this way, we are able to prove that $\mathcal{L}[\mathcal{O}]$ is $SU(D)$ for all physical representations $u$; see below.

\subsection{Graphical Description of Computational Power}

The primitive gates in our scheme that can be executed in a single step are generated by elements from the following set $\mathcal{O}$ of antihermitian opertors:

\begin{equation} \nonumber
\mathcal{O}=\left\{ \alpha C^{i\dagger}C^j-\alpha^*C^{j\dagger}C^i \right\} \ \ \forall i\neq j, \ \forall|\alpha|\ll 1
\end{equation}

\noindent Throughout the following, $\alpha$ always represents an arbitrary complex number of small magnitude, unless stated otherwise. By concatenating these primitive gates, we can execute any unitary gate generated by elements of the algebra $\mathcal{A}[\mathcal{O}]$ defined as the smallest Lie algebra containing $\mathcal{O}$. This algebra, called the dynamical Lie algebra in the context of quantum control, determines the computational power of the resource state; the set of gates $\mathcal{L}[\mathcal{O}]=e^{\mathcal{A}[\mathcal{O}]}$. We call our resource state universal if $\mathcal{L}[\mathcal{O}]=SU(N)$, such that $\mathcal{A}[\mathcal{O}]=su(N)$, the Lie algebra of traceless antihermitian matrices with commutator bracket. $su(N)$ can be spanned by the operators $\alpha X^i Z^j -\alpha^* Z^{j\dagger}X^{i\dagger}$ for $i,j=0,\dots, N-1$, and $\alpha\in\mathbb{C}$. Denote these operators by $O_{i,j}(\alpha)$. Clearly, by the definition of $\mathcal{O}$, we have $\mathcal{A}[\mathcal{O}]\subset su(N)$ always.

\begin{figure*}[htp]
  \centering
  \subfigure[\label{fig:grid1}]{\includegraphics[scale=0.31]{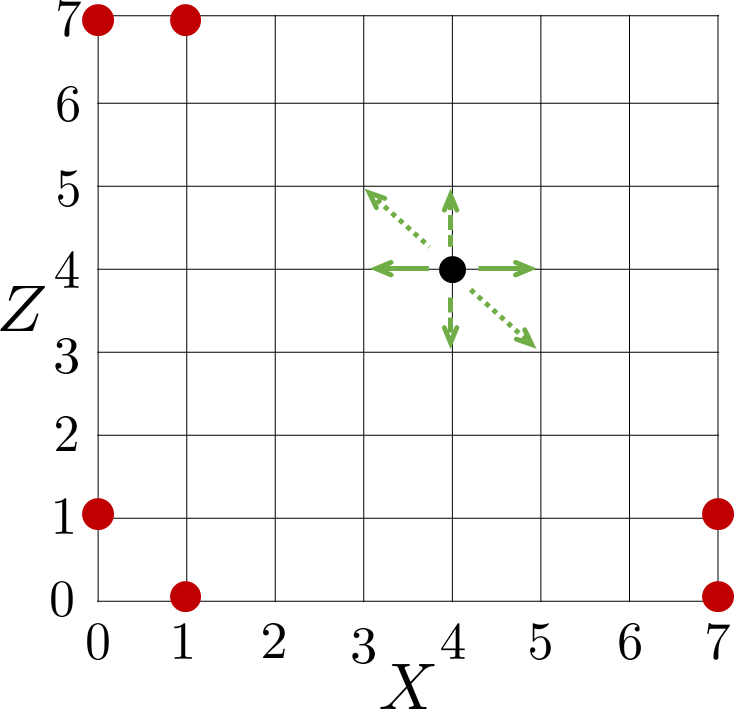}}\quad
  \subfigure[\label{fig:grid2}]{\includegraphics[scale=0.31]{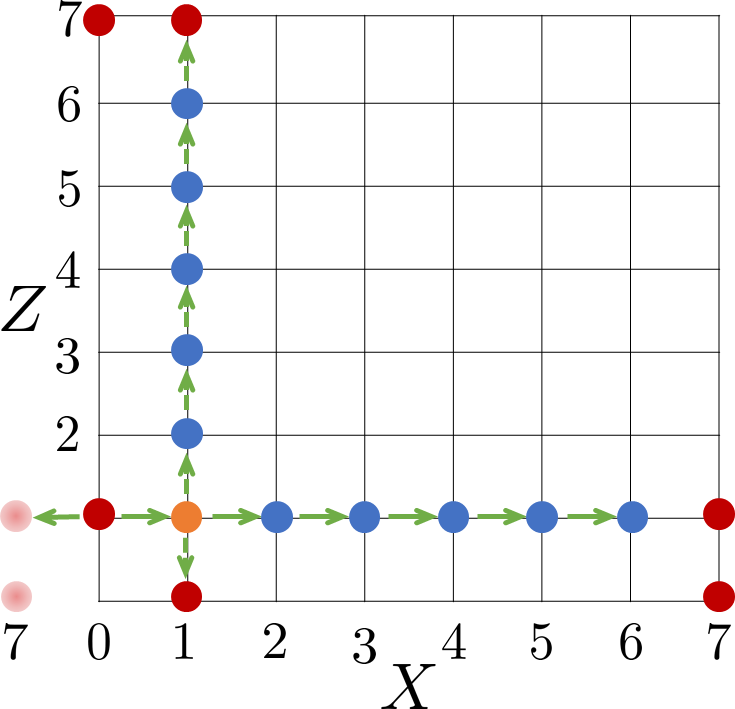}}\quad
  \subfigure[\label{fig:grid3}]{\includegraphics[scale=0.31]{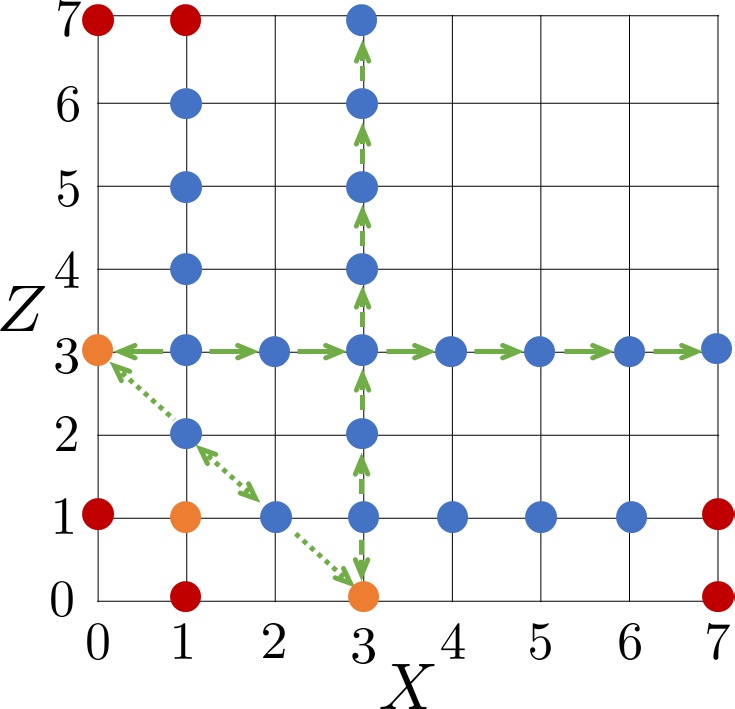}}\quad
  \subfigure[\label{fig:grid4}]{\includegraphics[scale=0.31]{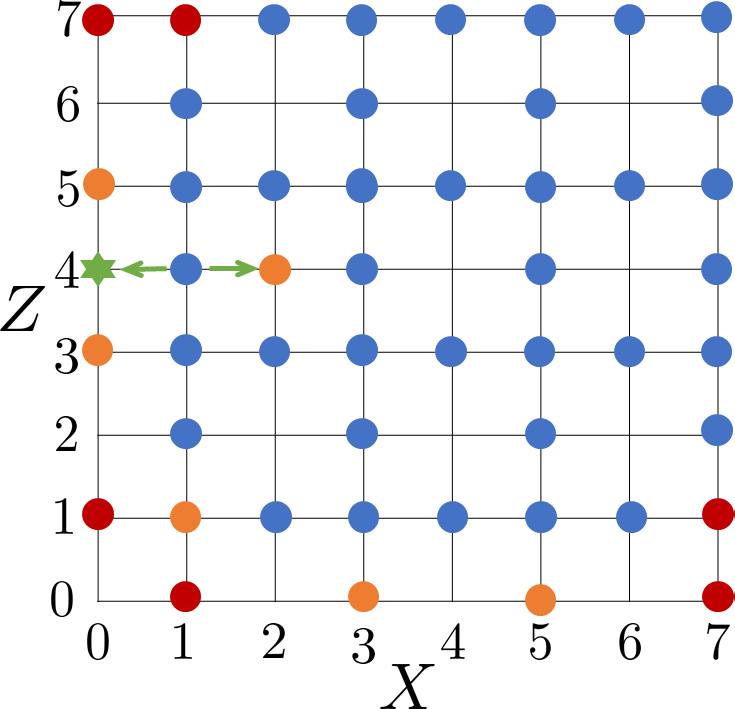}}
  \caption{Illustration of graphical description of $\mathcal{A}[\mathcal{O}]$ and proof of Lemma 3 for $D=8$ and $r=1$. (a) Initial points guaranteed by Lemma 2 (red marks). Arrows indicate basic moves: X move is solid, Z dashed, Y dotted. Starting from a marked point, if any arrow points to another marked point, its opposite can be marked as well. (b) First basic move is done from position (0,1), abusing periodic boundary conditions indicated by faded marks, creating orange mark. Starting from the orange mark, row/column 1 are filled using X/Z moves. (c) A Y move is used to obtain new starting points in row/column 3 (orange marks), which are then filled using X/Z moves. (d) After every other row/column is filled, the point (2,4) is filled using the special rule with the green star indicating the hermitian point (0,4). The remaining points can now be filled with basic moves.}
   \label{fig:grid}
\end{figure*}

The task is then, given $\mathcal{O}$, to determine $\mathcal{A}[\mathcal{O}]$. This is facilitated by a graphical interpretation. The elements of $\mathcal{A}[\mathcal{O}]$ can be indexed by a pair of mod $D$ integers $(i,j)$, which refer to the set of operators $\left\{O_{i,j}(\alpha)\right\}$ for all $\alpha\in\mathbb{C}$. We can construct a $D\times D$ grid, whose vertices correspond to pairs $(i,j)$. We place a marker on a vertex if the corresponding operators are in $\mathcal{A}[\mathcal{O}]$(See Fig.~\ref{fig:grid}). It is then clear that we have $\mathcal{A}[\mathcal{O}]=su(N)$ if and only if we can mark all vertices on the graph. Note that half of the points on the graph are redundant, since $(i,j)$ and $(D-i,D-j)$ refer to the same operators. So whenever $(i,j)$ is marked, we can mark $(D-i,D-j)$ for free.

Different physical representations $u$ determine the operators in $\mathcal{O}$, which in turn determines the initial conditions of our grid. We briefly pause for a Lemma that restricts the possible initial conditions, based on the finite correlation length of our state.

\begin{Lemma}
If $D=p^n$ is a prime power, then there exists an integer $r$ which is not divisible by $p$ such that:

\begin{equation} \label{eq:l1}
\left\{O_{1,0}(\alpha),O_{0,r}(\alpha),O_{D-1,r}(\alpha)\right\}\subset \mathcal{O}
\end{equation}

\end{Lemma}

The proof will be presented shortly. This lemma allows us to define the basic moves that can be used to fill up the graph starting from this initial point. Consider the following commutator of two elements in $\mathcal{A}[\mathcal{O}]$:

\begin{equation}
\left[O_{i,j}(\alpha),O_{1,0}(1) \right]=O_{i+1,j}(\alpha(\Omega^{-j}-1))-O_{i-1,j}(\alpha(\Omega^j-1))
\end{equation}

So, starting from the point $(i,j)$, we get out a linear combination of operators represented by points $(i-1,j)$ and $(i+1,j)$. If either of these points are already filled, we can simply subtract out that part we already have, allowing us to fill the other point since $\alpha$ is a free parameter. We can do the same thing with the operator $O_{0,r}(1)$ and also $O_{D-1,r}(1)$. So our three basic moves are, starting from a marked point $(i,j)$:

X) Inspect points $(i+1,j$) and $(i-1,j)$. If one is marked already and $j\neq 0$, mark the other.

Z) Inspect points $(i,j+r)$ and $(i,j-r$). If one is marked already and $i\neq 0$, mark the other.

Y) Inspect points $(i-1,j+r)$ and $(i+1,j-r)$. If one is marked already and $ir+j\neq 0$, mark the other.

See Fig.~\ref{fig:grid1} for an illustration. Recall that, since we are working with integers mod $D$, our graph has periodic boundaries. Basic moves are forbidden for certain values of $(i,j)$ because these correspond to taking the commutator of commuting operators, which will give 0. The points where a Y move are forbidden form a line with slope $-r$ starting from the origin, 

There is one final rule that applies only when $D$ is even: If at any time a point corresponding to a hermitian operator [i.e. $(\frac{D}{2},0),(0,\frac{D}{2}),(\frac{D}{2},\frac{D}{2})$] can be inspected by a basic X, Y, or Z move, it can be considered marked. This rule can be explained by an example. Consider the commutator:

\begin{equation}
\left[O_{1,\frac{D}{2}}(\alpha),O_{1,0}(\beta)\right]=O_{0,\frac{D}{2}}(2\alpha\beta^*)-O_{2,\frac{D}{2}}(2\alpha\beta)
\end{equation}

Since $Z^{\frac{D}{2}}$ is hermitian, we can choose $\alpha=\beta$ and annihilate the first term automatically, leaving the second term with the free coefficient $\alpha$. This allows us to mark $O_{2,\frac{D}{2}}(\alpha)$, which in turn allows us to mark $O_{0,\frac{D}{2}}(\alpha)$. This process generalises whenever one of the operators is hermitian, giving the basis for the final rule. With these rules in hand, we have our final lemma:

\begin{Lemma}
With the initial conditions of Lemma 2, each point on the graph can be marked using basic moves. That is, the set of operators in Eq.~\ref{eq:l1} generates $su(p^n)$.
\end{Lemma}

\begin{proof}
We prove first for the case $r=1$, and comment on general $r$ at the end. Based on Lemma 2, our initially marked points include $(1,0),(0,1),(D-1,1)$ where $D=p^n$. We also get $(D-1,0),(0,D-1),(1,D-1)$ by the aforementioned redundancy of the points (See Fig.~\ref{fig:grid1}) We start with an X move from $(0,1)$. Since $(D-1,1)$ is filled, we can fill $(1,1)$. Using a sequence of X/Z moves, we get $(i,1)/(1,i)$ for all $i$ (See Fig.~\ref{fig:grid2}). Now we perform a Y move from $(2,1)/(1,2)$ to get $(3,0)/(0,3)$. Again using a sequence of X/Z moves, we get $(i,3)/(3,i)$ for all $i$ (See Fig.~\ref{fig:grid3}). Continuing in this fashion, we can fill every other row and every other column. We must now separate the proof into two cases:

Case 1: $p$ odd $(p\neq 2)$. In this case, the periodic boundary conditions mean that filling every other row/column will in fact fill every row/column. So the above procedure is enough to fill the grid.

Case 2: $p=2$. Here, filling every other row/column will miss half of the rows/columns, so we must use the final rule involving hermitian operators to continue. We use an X move at $(1,\frac{D}{2})$. Since $(0,\frac{D}{2})$ corresponds to a hermitian operator, we mark $(2,\frac{D}{2})$ for free (See Fig.~\ref{fig:grid4}). It is now clear that we can mark all remaining point using basic moves. 

A final check is that we did not perform any forbidden moves in the above procedure; this can be easily verified. The case $r\neq 1$ is almost identical. Our initially marked points include  $(1,0),(0,r),(D-1,r)$. A key observation is that, since $r$ is coprime with $D$, all integers $0,\dots,D-1$ can be obtained as multiples of $r$. Then we can repeat the above procedure of filling rows and columns one by one. The only difference is the order in which they are filled. 
\end{proof}

By our two lemmas, we have $\mathcal{A}[\mathcal{O}]=su(p^n)$ in every case. Since $p^n$ was an arbitrary divisor of $|H|$, we have completed the proof of the theorem. We finish with the proof of Lemma 2.

\begin{proof} \textit{of Lemma 2}. 
It is convenient to assume that $C^0=I$. This can always be done by enacting a transformation $C^i\to \tilde{C}^i=C^{0\dagger}C^i$. Such a transformation does not change $\mathcal{O}$; it is just a relabelling of the elements. Define the set $\tilde{\mathcal{C}}=\left\{ \tilde{C}^i,i=0\dots d-1\right\}$. 

In order to impose further structure on $\tilde{\mathcal{C}}$, we use injectivity. Namely, the fact that our MPS is short-range correlated implies that the set of products $\{A^{i_1}A^{i_2}\dots A^{i_L}\}$ spans the space of all complex matrices for large enough $L$. By tracing out the junk subspace corresponding to the matrices $B^i$, we see that this property holds on the logical subspace alone. That is, every $D\times D$ matrix can be expressed as a linear combination of products $C^{i_1}C^{i_2}\dots C^{i_n}$. 

This property holds also for the matrices $\tilde{C}^i$. To see this, we use the fact that the spanning property is equivalent to \textit{primitivity} of the channel $C(X)=\sum_i C^i X C^{i\dagger}$, which means that $C^{\circ L} (X)$ has full rank for all $X$ \cite{Sanz2010}. Since all Heisenberg-Weyl operators commute up to a phase, we have $\tilde{C}^{\circ L} (X)=(C^{0\dagger})^L C^{\circ L }(X) (C^0)^L$ which shows clearly that $\tilde{C}(X)$ is primitive as well by the unitarity of $C^0$. 

Now, since the matrices $\tilde{C}^i$ span all matrices with their products, and their products are always Heisenberg-Weyl operators, they must generate the entire set of Heisenberg-Weyl operators, up to complex phases. This means we must have a pair of operators $\tilde{C}^a$ and $\tilde{C}^b$ such that $\tilde{C}^a\tilde{C}^b=\Omega^r \tilde{C}^b\tilde{C}^a$ where $p$ does not divide $r$. If not, define the numbers $r_{ij}$ by $\tilde{C}^i\tilde{C}^j=\Omega^{r_{ij}} \tilde{C}^j\tilde{C}^i$. Then a commutator of arbitrary elements can be written:

\begin{equation}
\prod_i (\tilde{C}^i)^{a_i}\prod_j (\tilde{C}^j)^{b_j}=\Omega^{\sum_{ij} a_i  b_j r_{ij}}\prod_j (\tilde{C}^i)^{b_j}\prod_i (\tilde{C}^i)^{a_i}
\end{equation}

If $p|r_{ij}$ for all $i,j$, then $p|\sum_{ij} a_i  b_j r_{ij}$, which cannot always be true if we have a generating set. Finally, for any unitary operators that commute like $\tilde{C}^a\tilde{C}^b=\Omega^r \tilde{C}^b\tilde{C}^a$ there exists a unitary $U$ such that $UC^aU^\dagger=X$ and $UC^bU^\dagger=Z^r$ (see Ref.~\cite{Zhou2003} for the case $r=1$, which generalizes easily). In the basis defined by $U$, we have $I,X,Z^r\in \tilde{\mathcal{C}}$, which gives the claimed operators in $\mathcal{O}$. 
\end{proof}

\bibliographystyle{apsrev4-1}
\bibliography{bibliography}

\end{document}